\documentclass[10pt,a4paper]{article}
\usepackage[english]{babel}
\usepackage{amsmath}
\usepackage{amsfonts}
\usepackage{amssymb}
\usepackage{color}
\usepackage{yfonts}
\usepackage{graphicx}
\usepackage{subcaption}
\usepackage{float}
\usepackage{enumerate}
\usepackage{dsfont}
\usepackage{subfig}
\usepackage{xcolor}
\usepackage{comment}
\usepackage{cancel}
\usepackage{ulem}
\usepackage{lscape}
\usepackage[round,authoryear]{natbib}
\newtheorem{theorem}{Theorem}

\newtheorem{definition}[theorem]{Definition}

\newtheorem{lemma}[theorem]{Lemma}

\newtheorem{remark}[theorem]{Remark}

\newenvironment{proof}{\paragraph{Proof:}}{\hfill$\square$}
\usepackage{media9}

\begin{document}
\title{
	Probability equivalent level for CoVaR and VaR in bivariate Student-\textit{t} copulas with application to foreign exchange risk monitoring}

\author{Daniela I. Flores-Silva (daniela.flores@gm.uca.es),\\  Miguel A. Sordo (mangel.sordo@uca.es),\\ Alfonso Su\'arez-Llorens 
(alfonso.suarez@uca.es)}

\date{September 02, 2025}
\maketitle

{\small{Departamento de Estadística e Investigación Operativa, Facultad de Ciencias, Universidad de Cádiz, Spain.}}

Version: September 02, 2025

\maketitle

\begin{abstract}
We extend the ``probability-equivalent level of VaR and CoVaR" (PELCoV) methodology to accommodate bivariate risks modeled by a Student-\textit{t} copula, relaxing the strong dependence assumptions of earlier approaches and enhancing the framework's ability to capture tail dependence and asymmetric co-movements. While the theoretical results are developed in a static setting, we implement them dynamically to track evolving risk spillovers over time. We illustrate the practical relevance of our approach through an application to the foreign exchange market, monitoring the USD/GBP exchange rate with the USD/EUR series as an auxiliary early warning indicator over the period 1999–2024. Our results highlight the potential of the extended PELCoV framework to detect early signs of risk underestimation during periods of financial stress.

\end{abstract}


\section{Introduction}
Value at Risk (VaR) and Conditional Value at Risk (CoVaR) are two widely used risk measures that have garnered increasing attention from both researchers and practitioners in recent years, due to their effectiveness in capturing tail risk and systemic interdependencies in financial markets. In this study, given a bivariate risk \((X, Y)\), we build upon these concepts by developing a unified framework known as the ``probability-equivalent level of VaR and CoVaR'' (PELCoV), originally introduced by Ortega-Jiménez et al.\ (2024) under the assumption that the underlying vector exhibits strictly stochastically increasing dependence.

We extend this framework to the case of a bivariate random vector governed by a Student-\textit{t} copula, which does not satisfy the aforementioned dependence property. While the theoretical analysis is conducted within a static framework, practical implementation requires a dynamic setting, as the dependence structure of the vector may vary over time.

Accordingly, we apply our findings to risk monitoring in the foreign exchange market, using bivariate time series models based on time-varying copulas. Specifically, we monitor the risk exposure associated with the USD/GBP exchange rate (i.e., the US dollar to British pound) by observing the auxiliary series USD/EUR (US dollar to euro), and analyzing the associated \(\text{PELCoV}\) over  the period from January 1, 1999, to April 1, 2024.

\subsection{Background}
    
    Value at Risk (VaR) is a widely used risk measure that quantifies the potential loss a financial institution or portfolio may face over a specified time horizon, under normal market conditions, and at a given confidence level. Formally, let $Y$ denote a random variable representing losses, with cumulative distribution function (CDF) $F_Y$. The VaR at probability level $v \in (0,1)$ is defined as:
	
	\[
	\text{VaR}_{v}[Y] = F_Y^{-1}(v) = \inf \left\{ x : F_Y(x) \geq v \right\}.
	\]
	This quantile-based measure identifies the loss threshold that will not be exceeded with probability $v$. Since its adoption by the Basel Committee on Banking Supervision (BCBS) in the 1990s, VaR has been extensively utilized in regulatory frameworks, including Basel III/IV for banking and Solvency II for insurance. For a comprehensive overview of the method and its applications in theory and practice, see  \citet{jorion2000risk}.  

A major limitation of Value at Risk (VaR) is its inability to account for the interdependencies between financial institutions, which are essential for assessing systemic risk. Systemic risk refers to the potential for a collapse or significant disruption in the entire financial system, rather than just individual institutions. This concept has been extensively studied in the literature (see   \citet{bisias2012survey} and  \citet{benoit2017risks} for comprehensive surveys). One of the primary channels through which systemic risk manifests is financial contagion, where losses in one institution spread to others due to their interconnected exposures (see  \citet{glasserman2016contagion}). Substantial research has been dedicated to quantifying, estimating,  and comparing different measures of contagion risk, including significant contributions from  \citet{chen2014asset},  \citet{girardi2013systemic},  \citet{mainik2014dependence},  \citet{sordo2015comparison,sordo2018stochastic},  \citet{tobias2016covar},  \citet{acharya2017measuring}, \citet{ortega2021stochastic,ortega2024}, \citet{beutner2024residual} and   \citet{francq2025inference}.

To address VaR’s shortcomings in capturing systemic risk, \citet{tobias2016covar}   introduced Conditional Value at Risk (CoVaR), an extension of VaR that explicitly incorporates interdependencies. CoVaR measures the risk of a financial institution conditional on the distress of another institution, offering a more comprehensive perspective on systemic vulnerabilities. Formally, the co-value-at-risk (CoVaR) of $Y$ at level $v \in \left( 0,1\right)$, given that $X$ is at level $u \in \left( 0,1\right)$, denoted as CoVaR$_{v, u}\left[ Y|X\right]$, is defined as the VaR of the conditional variable $\left[ Y|X=\text{VaR}_u\left[X\right]\right]$ at risk level $v$, as follows:
		\[
	\text{CoVaR}_{v, u}\left[ Y|X\right] = \text{VaR}_v\left[ Y|X=\text{VaR}_u\left[X\right]\right] = F^{-1}_{Y|X=\text{VaR}_u\left[X\right]}(v).
	\]
Although VaR and CoVaR offer different approaches to risk monitoring, they can be combined to improve risk assessment. In a recent study, \citet{ortega2024} explored a strategy that combines conditional and unconditional VaR, investigating the conditions that establish the ordering between VaR and CoVaR. For a random variable $Y$ representing financial risk, this strategy requires the presence of a covariate $X$ whose dependence structure with $Y$ is easily observable,  although not necessarily strong, to ensure that $X$ effectively contributes to monitoring the risk of $Y$.  The concept is straightforward: suppose that, initially, the risk is monitored using VaR at level \( v \) of \( Y \), for some \( v \in (0,1) \). Assume that \( X \) has reached a risk level \( u \) such that \( \text{CoVaR}_{v,u}[Y|X] > \text{VaR}_v[Y] \). A prudent investor who prioritizes minimizing potential losses and safeguarding capital -even at the cost of potentially lower returns- such as institutional investors, risk-averse individuals, or regulatory bodies, should recognize that VaR underestimates the spillover effect. In this case, it would be prudent to replace $\text{VaR}_v[Y] $ with \( \text{CoVaR}_{v,u}[Y|X] \) to provide a more cautious risk assessment. Given \( v \in (0,1) \), the approach involves determining the set of risk levels  
\begin{equation}\label{setA}  
    A(v) = \{ u_v \in (0,1) : \text{CoVaR}_{v,u_v}[Y|X] = \text{VaR}_v[Y] \}  
\end{equation}  
and analyzing the relative order of VaR and CoVaR in the intervals between consecutive points of \( A(v) \). Each \( u_v \in A(v) \) represents a probability level at which CoVaR and VaR are equal; therefore, the focus is on understanding how their relationship changes in the intermediate regions.

The following definition, taken from \citet{ortega2024} formalizes this concept. 
\begin{definition}\label{pelcovdef} Let $(X,Y)$ be a random vector and let   $v\in(0,1)$. A probability equivalent level of CoVaR-VaR at risk level $v$  (PELCoV$_v$) for $X$ is any $u_v\in (0,1)$ that satisfies $\text{CoVaR}_{v,u_v}\left[ Y|X\right]=\text{VaR}_v\left[ Y\right]$.
\end{definition}
For mathematical tractability, we assume throughout this paper that the random vector \( (X,Y) \) has absolutely continuous, strictly increasing marginal distribution functions \( F_X(\cdot) \) and \( F_Y(\cdot) \), as well as a strictly increasing conditional distribution \( F_{Y|X=x}(\cdot) \) for all \( x \), defined on their respective supports. We refer to these properties as the regularity conditions. \citet{ortega2024} proved that, under these regularity conditions, a PELCoV$_v$ depends solely on the copula $C$ of the random vector $(X,Y).$  To formalize this, we recall the definition of a copula. According to Sklar's theorem, if $K$ is the joint distribution of the random vector $(X,Y),$  we can express it as
\begin{equation*}
	K(x, y)=C\left(F_X(x), F_Y(y)\right),
\end{equation*}
where $C$ is the copula, the joint distribution function of the vector $(U,V),$ with $U=F_X\left( X \right)$ and $V= F_Y\left(Y\right)$. The copula function \( C \) captures the dependence structure between the components of the vector, independent of their marginal distributions. Under the regularity conditions, \( C \) is unique and differentiable (see  \citet{nelsen2006introduction}).

The following theorem summarizes key properties of PELCoV$_v$, as established by  \citet{ortega2024}.
\begin{theorem}\label{basic} Let $\left( X, Y\right)$ be a random vector satisfying the regularity conditions  with copula $C$ and let $v \in (0,1).$ Then, \newline
(a) CoVaR$_{v,u}[Y\mid X]\ge  \text{VaR}_v[Y]$ (respectively $\le,=$) if, and only if  $\partial_1 C(u,v)\le v$ (respectively $\ge,=$). \newline
(b) $CoVaR_{v,u}[Y\mid X]$  is continuous in $u\in (0,1)$ if, and only if,  $\partial_1 C(u,v)$ is continuous in $u\in (0,1)$. \newline
(c) If $\partial_1 C(u,v)$ is continuous in $u\in (0,1)$, then there exists at least one $u_{v}\in (0,1)$ such that  CoVaR$_{v, u_v}[Y| X] = \text{VaR}_v[Y]$.
\end{theorem}
\begin{remark}\label{nota}
According to Part (a) of Theorem \ref{basic}, for any \( v \in (0,1) \), the elements of \( A(v) \) or PELCoV\(_v\)s are the solutions to $\partial_1 C(u,v) = v.$\end{remark}

\subsection{Motivation}

Assume that the copula \( C \) of the vector \( (X,Y) \) satisfies that \( \partial_1 C(u,v) \) is continuous for all \( u \in (0,1) \), so that, by part (c) of Theorem \ref{basic}, the set \( A(v) \) defined in \eqref{setA} is nonempty. \citet{ortega2024} showed that under an additional positive dependence property between \( X \) and \( Y \), known as \textit{Strict Stochastically Increasing}\footnote{The Strictly Stochastically Increasing (SSI) property is a slight modification of the Stochastically Increasing (SI) property, also known as Positive Regression Dependence (PRD), a concept introduced by   \citet{lehmann1966theorem}. The SI concept does not require the growth of the conditional probability to be strict.} (SSI), the set \( A(v) \) defined in \eqref{setA} is a singleton. Recall that \( Y \) is said to be SSI in \( X \), denoted \( Y \uparrow _{SSI} X \), if the survival probability \( \text{Pr}\{Y > y \mid X = x\} \) is a strictly increasing function of \( x \), for all \( y \).
Intuitively, if \( Y \uparrow _{SSI} X \), we expect \( Y \) to take large values as the conditional random variable \( X \) increases. The SSI property is characterized by the copula: \( Y \uparrow _{SSI} X \) is equivalent to the condition that the partial derivative
\[
\partial_1 C(u,v) = \text{Pr}\{ V \leq v \mid U = u \}
\]
is a strictly decreasing function of \( u \), for all \( v \), where \( (U,V) \) is defined as above. Under the assumption that \( Y \uparrow _{SSI} X \), the probability level \( u_v \) serves as an alert system, indicating when VaR begins to underestimate the risk relative to CoVaR. Using this methodology,  \citet{ortega2024} derive and interpret the PELCoV$_v$ for various copula families that satisfy the SSI property for positive values of their dependence parameters, including the bivariate Gaussian, Farlie-Gumbel-Morgenstern, Frank, Clayton, and Ali-Mikhail-Haq copulas (see  \citet{nelsen2005copulas}, for formulas and further details on these copulas).

This approach is particularly relevant in financial econometrics, where a fundamental method for modeling relationships between positively dependent random variables is the classic regression framework:  
\begin{equation}\label{regression}  
    Y = \phi(X) + \sigma \varepsilon,  
\end{equation}  
where \( \phi: \mathbb{R} \to \mathbb{R} \) is a strictly increasing function, and \( \varepsilon \) represents random noise with mean zero and unit variance, independent of \( X \). This model describes how the response variable \( Y \), often representing financial quantities such as asset returns, volatility measures, or risk premia, evolves as a function of the explanatory variable \( X \) while incorporating stochastic fluctuations. In this setting, the stochastic monotonicity property \( Y \uparrow_{SSI} X \) holds trivially.  Moreover, according to Proposition 9 in \citet{ortega2024}, the PELCoV\(_v\), given by  
\[
u_v = F_X\left( \phi^{-1} \left(\text{VaR}_v[Y] - \sigma \text{VaR}_v[\varepsilon] \right) \right),
\]  
is increasing with respect to \( v \in (0,1) \) whenever \( \varepsilon \) has a log-concave density function. This condition includes, in particular, the case where \( \varepsilon \) follows a normal distribution.  However,  any slight modification of model \eqref{regression}, such as  
\begin{equation*}  
    Y = \phi(X) +  \sigma(X) \varepsilon,  
\end{equation*}  
where   \( \sigma: \mathbb{R} \to \mathbb{R^+} \) is an increasing function,  can cause the vector \( (X,Y) \) to no longer satisfy the property \( Y \uparrow_{SSI} X \). This occurs, for example, when \( \varepsilon \) follows a normal distribution, since the conditional variable  
\[
\{Y \mid X = x\}=\phi(x)+\sigma(x)\varepsilon  \sim  N(\phi(x),\sigma(x))  
\]  
does not satisfy the SSI property unless \( \sigma(x) \) is constant for all \( x \). This observation emphasizes the necessity of exploring broader approaches beyond SSI-based methodologies for applying
PELCoV\(_v\).

\subsection{Aim of the paper}
A crucial dependence structure in the econometric analysis of financial time series is the Student-\textit{t} copula. Unlike the Gaussian copula, the Student-\textit{t} copula provides non-zero tail dependence, making it a superior tool for modeling financial markets, which often experience extreme co-movements during periods of turmoil. Its ability to capture joint tail risk is essential for rigorous risk management and portfolio modeling. Indeed, the Student-\textit{t} copula remains a popular parametric choice in risk management and financial econometrics, as highlighted by \cite{ShyamalkumarTao2022}, who explore its effectiveness in modeling multivariate financial return data. \cite{ShimLee2017} further demonstrate how integrating the Student-\textit{t} copula with a GARCH framework accommodates skewness, heavy tails, volatility clustering, and evolving conditional dependencies in financial time series. More recently, \cite{FILIPIAK2025105490} provide evidence of the continued relevance and practical advantages of the Student-\textit{t} copula in modern financial econometrics.

The cumulative distribution function of the univariate Student-\textit{t}-distribution with \( n \) degrees of freedom is given by:
\[
t_n(x) =\int_{-\infty}^{x} \frac{\Gamma\left(\frac{n+1}{2}\right)}{\sqrt{n\pi} \, \Gamma\left(\frac{n}{2}\right)} \left(1 + \frac{s^2}{n}\right)^{-\frac{n+1}{2}}ds, \ \ x \in \mathbb{R},
\]
where \( \Gamma(\cdot) \) denotes the Gamma function.  The bivariate Student-\textit{t} copula, for $(u,v)$ in $[0,1]^2,$  is defined as
$$
C_n(u,v)=\int_{-\infty}^{t_n^{-1}(u)}\int_{-\infty}^{t_n^{-1}(v)}\frac{1}{2\pi\sqrt{1-\rho^2}}\left(  1+\frac{s_1^2+s_2^2-2\rho s_1 s_2}{n(1-\rho^2)}  \right)^{-\frac{n+2}{2}} ds_1ds_2, $$
where \( n>1 \) and \( \rho \in (-1,1) \) are the copula parameters. 
A random vector $(X,Y)$ with a bivariate Student-\textit{t} copula exhibits positive quadrant dependence\footnote{Note that a vector $(X,Y)$ is PQD (respectively, SI) if and only if its copula is PQD (respectively, SI). See Theorem 3.10.19 in  \citet{muller2002comparison} and the discussion in \citet{cai2012invariant}. } (PQD) when $\rho>0,$ meaning that 
$$ P[X>x,Y>y]\ge P[X>x]P[Y>y], \text{ for all } x,y \in \mathbb{R}.$$ 
Intuitively, this indicates that $X$ and $Y$ are more likely to attain large values simultaneously than if they were independent with the same marginal distributions. However, $(X,Y)$  does not satisfy the property SSI, as the conditional scale of the bivariate $t_n$ diverges as $x \to \pm \infty.$ For a detailed discussion of these properties, see \citet{joe2014dependence} , p. 182. Figure \ref{notSI} illustrates that the function $v\to\partial_1C(u,v)$ does not strictly decrease in \( u \), for all \( v \), which is a necessary condition for the SI property.

 \begin{figure} 
\centering
\includegraphics[scale=0.5]{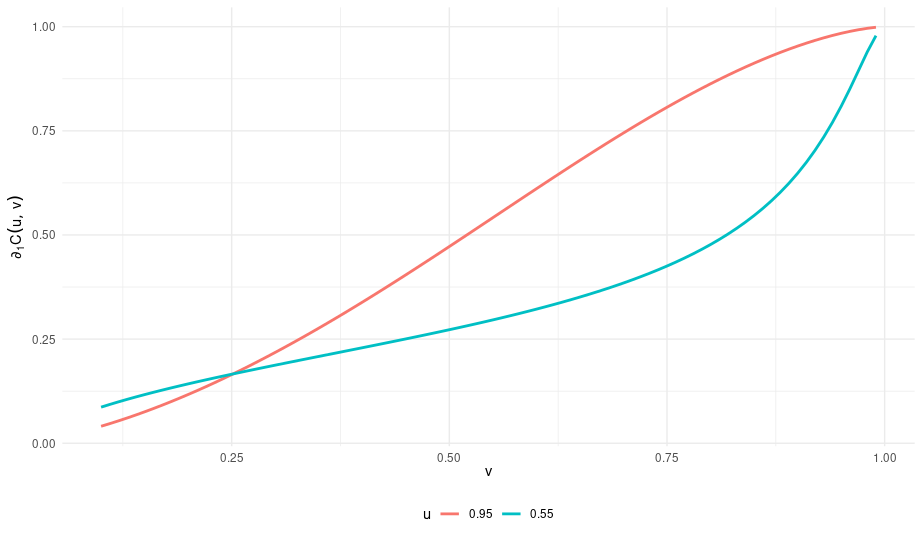}
\caption{Function $v\to\partial_1C(u,v)$, where $C(u,v)$ is the Student \textit{t}-copula with parameters $\rho=0.4, n=2$ and  fixed values of $u=0.95,0.55$. }\label{notSI}
\end{figure}

It is therefore relevant and the objective of this study, to address the following question regarding PELCoV\(_v\) in the case of a bivariate random vector whose dependence structure is governed by a Student-\textit{t} copula: How many elements are in the set \( A(v) \), and how can this information be used for risk assessment and monitoring in financial contexts? To illustrate the practical implications of our findings, we apply them to the risk assessment in the foreign exchange market
using bivariate time series modeled with time-varying Student-\textit{t} copulas. Following the approach of  \citet{patton2006modelling}, we account for potential dynamic dependence structures over time, assuming that the copula remains a Student-\textit{t} copula throughout the study period while its parameter $\rho$ evolves according to a specified evolution equation. Consequently, this procedure will naturally yield a 
PELCoV$_v$ that also varies over time.

\section{PELCoV$_v$'s in Student-\textit{t} copulas}
  Let $Y$ be a random variable representing the returns or losses of a financial asset or risk. A supervisor monitors this risk using $\text{VaR}_{v}[Y],$ where $v$ is typically set at levels such as $0.95$ or $0.99$, either to mitigate extreme losses or to comply with financial regulations. The supervisor adjusts the investment strategy whenever $Y$ reaches the VaR threshold, which is determined based on the historical evolution of the asset.  

A strategy based on \(\text{PELCoV}_v\) first analyzes \(\text{CoVaR}_{u,v}[Y|X]\) in comparison to \(\text{VaR}_{v}[Y]\) for all \( u \in (0,1) \), considering a given risk factor \( X \) whose dependence structure with \( Y \) is well-defined through a copula \( C \) that satisfies the regularity conditions.  According to Remark \ref{nota}, the elements of $A(v),$ (i.e., the probability equivalent levels of VaR and CoVaR, given \( v \)), are the solutions to the equation  \(
\partial_1 C(u,v) = v.
\)  
The relative positions of \(\text{VaR}_{v}[Y]\) and \(\text{CoVaR}_{u,v}[Y|X]\) for values of \( u \) lying between successive \(\text{PELCoV}_v\) levels determine which of the two measures is the more conservative for each \( u \in (0,1) \). In this section, we examine this problem when \( C \) is a Student-\( t \) copula with \( \rho > 0 \).
\begin{lemma}\label{lemita}
      Let \((X, Y)\) be a random vector following a Student-\(t\) copula characterized by parameters \(\rho > 0\) and \(n> 1 \).     The function $h(u)= \partial_1 C(u,v)$ for all $u\in (0,1)$ satisfies the following properties.\newline
          (a) The function \( h(u) \) is given by 
      \begin{equation}\label{hu}
h(u)=t_{n+1}\left(\dfrac{t_{n}^{-1}\left(v\right)-\rho t_{n}^{-1}\left(u\right)}{\sqrt{\dfrac{(n+(t_{n}^{-1}(u))^2)(1-\rho^2)}{n+1}}}\right).    
 \end{equation}
      (b) For any \(v \in \left(0, 1\right)\), we define
      \begin{equation}\label{u*}
          u^*=t_n\left(\frac{-\rho n}{t_n^{-1}(v)}\right).
      \end{equation}If \( v > \frac{1}{2} \), then \( u^* < \frac{1}{2} \) and \( u^* \) corresponds to a maximum of \( h \). Conversely, if \( v < \frac{1}{2} \), then \( u^* > \frac{1}{2} \) and \( u^* \) corresponds to a minimum of \( h \). \newline
      (c) The function \( h(u) \) has the following limits at the boundaries of the domain:
    \begin{align}
    \lim_{u\to 0^+}h(u) &= t_{n+1}\left(\frac{\rho\sqrt{n+1}}{\sqrt{1-\rho^2}}\right)=L_0    \\
    \lim_{u\to 1^-}h(u) &= t_{n+1}\left(\frac{-\rho\sqrt{n+1}}{\sqrt{1-\rho^2}}\right)=L_1=1-
    L_0 \nonumber
\end{align}  
  \end{lemma}
  \begin{proof}
Part (a) follows directly from straightforward differentiation\footnote{Alternatively, equation \eqref{hu} can be verified in Appendix C.2 of \citet{aas2009pair}, where the authors derive the partial derivative of the Student \textit{t}-copula with respect to its second argument, \(\partial_2 C(u,v)\). By symmetry, the derivative with respect to the first argument, \(\partial_1 C(u,v)\), follows analogously}.  To prove part (b), we note that the equation \( \frac{d}{du} h(u) = 0 \) holds if and only if
\(
-\rho n - t_n^{-1}(u) t_n^{-1}(v) = 0
\)
(see Appendix A for details). This condition determines the unique critical point \( u^* \) of \( h(u) \), given by \eqref{u*}.  The rest follows easily.  Finally, part (c) is straightforward to prove.
\end{proof}

It is insightful to compare the limits \( L_0 \) and \( L_1 \) in Lemma \ref{lemita} with the lower and upper tail dependence indices, respectively. These indices quantify the probability of joint extreme values in a bivariate distribution, which is crucial for risk management. Formally, given a bivariate random vector \( (X, Y) \) with a joint cumulative distribution function \( K \) and marginal distributions \( F_X \) and \( F_Y \), the lower tail dependence index, \( \lambda_L \), is defined by:
\[
\lambda_L = \lim_{u \to 0^+} P\left( Y \leq F_Y^{-1}(u) \mid X \leq F_X^{-1}(u) \right).
\]
whereas the upper tail dependence index, 
\( \lambda_U \), is defined as:
\[
\lambda_U = \lim_{u \to 1^-} P\left( Y > F_Y^{-1}(u) \mid X > F_X^{-1}(u) \right).
\]
These indices are directly derived from the copula that governs the dependence structure of the joint distribution. The general expression for the lower tail dependence index is 
\begin{equation}\label{lambda}
	\lambda_L = \lim_{u \to 0^+} \frac{C(u, u)}{u},
\end{equation}
while the upper tail dependence index is given by  
\[
\lambda_U = \lim_{u \to 1^-} \frac{1 - 2u + C(u, u)}{1 - u}.
\]
The Student-\( t \) copula exhibits both upper and lower tail dependence, meaning that extreme co-movements occur in both directions. This makes it particularly suitable for modeling financial returns, where crises often lead to strong dependence in both market downturns and upturns. For a Student-$t$ copula with \( n \) degrees of freedom and correlation \( \rho \), both indices are equal and given by  
\[
 \lambda_L =\lambda_U = 2t_{n+1} \left( -\sqrt{\frac{(n+1)(1-\rho)}{1+\rho}} \right) = 2t_{n+1} \left( -(1-\rho) \sqrt{\frac{n+1}{1-\rho^2}} \right),
\]
The limits \( L_0 \) and \( L_1 \) of the function $h(u)= \partial_1 C(u,v),$ which is the key function in the study of the PELCoV$_v$ for a given $v$, are given by
\begin{eqnarray}\label{lo}
L_0  & = &\lim_{u \to 0^+} P\left( Y \leq F_Y^{-1}(v) \mid X \leq F^{-1}(u) \right) \nonumber \\ & = & \lim_{u \to 0^+} \frac{C(u, v)}{u} \\ & = & \lim_{u \to 0^+} C_1(u,v) \nonumber
\end{eqnarray}
and 
\begin{eqnarray*}
1-L_1  & = &\lim_{u \to 1^-} P\left( Y > F_Y^{-1}(v) \mid X > F^{-1}(u) \right) \\ & = & \lim_{u \to 1^-} 1-\frac{v-C(u,v)}{1-u} \\ & = & 1-\lim_{u \to 1^-} C_1(u,v).
\end{eqnarray*}
These limits describe the asymptotic behavior of the conditional probability function at the boundaries of its domain. Comparing \eqref{lambda} and \eqref{lo}, the difference between \( L_0 \) and \( \lambda_L \) becomes clear. While \( L_0 \) quantifies the probability that \( Y \) falls below its risk threshold \( F_Y^{-1}(v) \) for a fixed \( v \), given that \( X \) takes extremely low values, \( \lambda_L \) measures the probability that \( Y \) also takes extremely low values given that \( X \) does. $L_1$ and $\lambda_U$ have a similar interpretation. 
 
\vspace{0.5cm}

Given $v\in (0,1)$, studying the elements of \( A(v) \) reduces to analyzing the solutions of the equation \( h(u) = v \).  
The case \( A\left(\frac{1}{2}\right) \) is particularly simple. Using \eqref{hu}, the equation \( h(u) = \frac{1}{2} \) reduces to:  
\[
t_{n+1} \left(\dfrac{-\rho \, t_{n}^{-1}(u)}{\sqrt{\dfrac{(n + (t_{n}^{-1}(u))^2)(1 - \rho^2)}{n+1}}} \right) = \dfrac{1}{2}.
\]
The unique solution to this equation is \( u = \frac{1}{2} \), implying that \( A\left(\frac{1}{2}\right) = \frac{1}{2} \).

 We now study the elements of \( A(v) \) for \( v \neq \frac{1}{2} \). The following lemma will be instrumental in the subsequent analysis. 

\begin{lemma}\label{disp}
 For $n>1,$ the following inequalities hold: \newline
 (a) If $v \in [0,\frac{1}{2}),$ then $ t_{n}^{-1}(v)\leq t_{n+1}^{-1} (v)$.  \newline 
 (b) If $v \in [\frac{1}{2},1),$ then $ t_{n+1}^{-1}(v)\leq t_{n}^{-1} (v).$ 
\end{lemma}
\begin{proof}
As established in  \citet{arias2005multivariate}, for any $m\ge n,$  the following inequality holds: 
 $$t_{m}^{-1}(p)-t_{m}^{-1}(q)\leq t_{n}^{-1}(p)-t_{n}^{-1}(q), \text{ for all } 0<q<p<1.$$
 Setting $m=n+1,$ we obtain part (a) by choosing $p=\frac{1}{2}$ and $q=v$ with $v<\frac{1}{2}$ and part (b) by choosing  $p=v$ with $v>\frac{1}{2}$ and $ q=\frac{1}{2}.$
 \end{proof}
\begin{theorem}\label{main_theorem}
Let \((X, Y)\) be a random vector with a Student-\(t\) copula characterized by parameters \(\rho > 0\) and \(n > 1\), and let \(v \in \left(\frac{1}{2}, 1\right)\). Then:
\begin{itemize}
    \item[(a)] The set \( A(v) \) necessarily contains a unique PELCoV\(_v\), denoted as \( u_{v1} \), within the interval \( \left(\frac{1}{2}, v^*\right) \), where \( v < v^* = t_n\left(\frac{t_n^{-1}(v)}{\rho}\right) < 1 \).
    \item[(b)] Additionally, a second PELCoV, denoted as \( u_{v2} \), exists if and only if \( v > L_0 \). If \( u_{v2} \) exists, it lies within the interval \( (0, u^{*}) \), where \( u^* \in \left(0, \frac{1}{2}\right) \) is defined by \eqref{u*}.
    \item[(c)] CoVaR\(_{u,v}[Y|X]\) is strictly less than VaR\(_v[Y]\) if and only if \( u \in (a, u_{v1}) \), where \( a = u_{v2} \) if the second PELCoV\(_v\) exists, and \( a = 0 \) otherwise.
\end{itemize}
\end{theorem}

\begin{proof}
      Let \( v \in \left(\frac{1}{2},1\right) \)  and define \( v^* = t_n \left( \frac{1}{\rho} t_n^{-1}(v) \right) > v \). To prove part (a), we begin by recalling from Lemma 4(b) that the function $h$ has a unique local maximum at $u^* < \frac{1}{2}$.
      Consequently, we can ensure that
      \( h(u) \) is strictly decreasing for $u$ in \( \left( \frac{1}{2}, v^* \right) \).
From equation \eqref{hu} we obtain  
\begin{eqnarray}
h\left(\frac{1}{2}\right)&=&t_{n+1}\left(\dfrac{t_{n}^{-1}\left(v\right)}{\sqrt{\dfrac{n(1-\rho^2)}{n+1}}}\right) \nonumber \\
&=& t_{n+1}(at_{n}^{-1}\left(v\right))  \nonumber \\
&\overset{(i)}>& t_{n+1}(t_{n}^{-1}\left(v\right)) \notag \\
&\overset{(ii)}\ge&v.  \label{uve}
\end{eqnarray}
Here, (i) follows from the fact that $a={\sqrt{\frac{n+1}{n(1-\rho^2)}}}>1,$ which increases the argument of  $t_{n+1}$ and thereby its value. Inequality (ii) follows from Lemma \ref{disp}(b).
Additionally, equation \eqref{hu} shows that $h(v^*)=\frac{1}{2}.$ Since \( h(u) \) is a continuous and strictly decreasing function in \( \left( \frac{1}{2}, v^* \right) \), with $h(\frac{1}{2})\ge v$ and $h(v^*)<v,$  it follows that the equation $h(u)=v$ has a unique solution,  $u_{v1},$ in the interval $(\frac{1}{2},v^*).$ In other words: the set \( A(v) \) necessarily contains a unique PELCoV\(_v\) within the interval \( \left(\frac{1}{2}, u^*\right) \), as stated in part (a).

To prove part (b), observe that the function $h$ is decreasing over the entire interval   \( (v^*,1) \). Given that $h(v^*)=\frac{1}{2}<v $ it follows that no additional solution can exist within the interval $ (v^*,1)$, meaning that if another solution exists, it must lie in the interval \( (0, \frac{1}{2}) \). 

We now proceed with the `if' part of the existence characterization. Suppose that \( v > L_0 \). Since \( v > \frac{1}{2} \), Lemma \ref{lemita} (b) ensures that \( u^* < \frac{1}{2} \) is a maximum of the function \( h \), which implies \( h(u^*) \geq h(u_{v1}) = v \). By the continuity of \( h \), we conclude that there must be a solution in the interval \( (0, u^*) \).  For the `only if' part, assume that a second PELCoV\(_v\), denoted by \( u_{v2} \), exists within the interval \( (0, \frac{1}{2}) \). By definition, this means \( h(u_{v2}) = v \). Recall that \( u^* \) is the unique critical point of \( h \) and corresponds to a maximum. Consequently, \( h \) is strictly increasing on \( (0, u^*) \) and strictly decreasing on \( (u^*,1) \). Furthermore, from \eqref{uve}, we know that \( h\left(\frac{1}{2}\right) > v \), which necessarily implies \( v > L_0 \), completing the proof of part (b). 

To prove part (c), note that \( u^* < \frac{1}{2} < u_{v1} \). If \( u_{v2} \) exists, then necessarily \( u_{v2} \in (0, u^*) \). Since \( h \) is a continuous function, strictly increasing on the interval \( (a, u^*) \) and strictly decreasing on the interval \( (u^*, u_{v1}) \), it is easy to analyze the sign of the function \( h(u) - v \) for all \( u \in (0,1) \). The sign analysis and the application of Theorem \ref{basic}(a) prove the result.
\end{proof}

\begin{remark}
Given \( v \in \left(\frac{1}{2}, 1\right) \), Theorem \ref{main_theorem} states that there is at least one PELCoV\(_v\) and at most two. The existence of a second PELCoV depends on the relative position of \( v \) with respect to \( L_0 \); in particular, a sufficiently high value of \( \rho \) guarantees that a second PELCoV\(_v\) does not exist.  However, in this case, the only remaining PELCoV\(_v\),  may become ineffective as an alarm signal, as it could remain too close to \( v \). For moderate values of \( \rho \), the presence of a second PELCoV\(_v\) in the left tail results from the extreme volatility in the tails of the Student-\( t \) distribution.

\end{remark}

The following theorem refers to the set \( A(v) \) for \( v < \frac{1}{2} \). Since the proof is similar to that of Theorem \ref{main_theorem}, it is omitted. 

\begin{theorem} 
Let \((X, Y)\) be a random vector with a Student-\(t\) copula characterized by parameters \(\rho > 0\) and \(n > 1\), and let \(v \in \left(0,\frac{1}{2} \right)\). Then:
\begin{itemize}
    \item[(a)] The set \( A(v) \) necessarily contains a unique PELCoV\(_v\), denoted as \( u_{v1} \), within the interval \( \left(v^*,\frac{1}{2}\right) \), where \( v > v^* = t_n\left(\frac{t_n^{-1}(v)}{\rho}\right)\).
    \item[(b)] Additionally, a second PELCoV, denoted as \( u_{v2} \), exists if and only if \( v < L_1 \). If \( u_{v2} \) exists, it lies within the interval \( ( u^{*},1) \), where \( u^* \in \left( \frac{1}{2},1\right) \) is defined by \eqref{u*}.
    \item[(c)] CoVaR\(_{u,v}[Y|X]\) is strictly less than VaR\(_v[Y]\) if and only if \( u \in ( u_{v1},b) \), where \( b = u_{v2} \) if the second PELCoV\(_v\) exists, and \( b = 1 \) otherwise.
\end{itemize}
\end{theorem}
To conclude this section, the following result demonstrates how to analytically determine the  PELCo$_v$s in Student t-copulas (see the appendix for details). 

\begin{lemma}  \label{ttkm}
Each \( u_v \in A(v) \)  satisfies the following equation:\begin{equation*}\label{PELCoV_solution}
u_v=t_{n}\left(\frac{t_{n}^{-1}(v)\rho(n+1)\pm \sqrt{k}}{\left(\rho^2(n+1)-\left(t_{n+1}^{-1}(v)\right)^2(1-\rho^2)\right)}\right)
\end{equation*}
where $$ k=\left((1-\rho^2)\left(t_{n+1}^{-1}(v)\right)^2\right)
\left( \rho^2(n+1)n^2+\left(t_{n}^{-1}(v)\right)^2(1+n\rho^2)\right)\ge 0.$$
\end{lemma}

\section{Application to Exchange Rate Risk Monitoring}  

Understanding the relationship between different exchange rates is crucial for analyzing the dynamics of the foreign exchange market and its economic implications. In particular, the link between the USD/EUR and USD/GBP exchange rates is of significant interest due to the central role these currencies play in international trade, investment, and monetary policy.  The euro (EUR) and the British pound (GBP) are among the most traded currencies globally, and their interaction with the US dollar (USD) is fundamental to financial market stability. However, structural differences between these currencies can influence their respective exchange rates against the USD. While the EUR is generally considered more stable and less volatile due to its backing by the Eurozone, the GBP has exhibited greater sensitivity to political and economic events, such as Brexit.

The previous analysis suggests a risk-monitoring strategy based on  PELCoV$_v$ methodology to assess exposure to the USD/GBP exchange rate ($Y_t$). This copula based approach integrates the unconditional VaR for USD/GBP with the CoVaR of USD/GBP given the USD/EUR exchange rate ($X_t$), which serves as an auxiliary and control variable.  
Copula-based approaches have been widely used to study exchange rate interdependencies, as demonstrated in the works of  \citet{aas2009pair},   \citet{loaiza2015exchange}, \citet{albulescu2018monetary}, \citet{liu2020forecasting} and  \citet{gong2022asymmetric}, among others. Several studies suggest that Student’s \textit{t} copula provides a better fit than alternatives such as Gaussian and Gumbel copulas to capture exchange rate dependencies, mainly due to its ability to model tail dependence and extreme co-movements (see  \citet{chen2004simple},  \citet{diks2010out} and   \citet{du2017copula}).

The analysis examines two key exchange rate time series: the US dollar to euro spot exchange rate (USD/EUR), which represents the price of one euro in terms of US dollars in the spot market, and the US dollar to UK pound sterling spot exchange rate (USD/GBP), which reflects the price of one UK pound sterling in terms of US dollars in the spot market\footnote{In the Federal Reserve Economic Data (FRED) database, the code for the U.S. dollar to euro spot exchange rate is EXUSEU, and the code for the U.S. dollar to U.K. pound sterling spot exchange rate is EXUSUK.}. The dataset, obtained from the Federal Reserve Economic Data (FRED) of the Federal Reserve Bank of St. Louis, covers the period from January 1, 1999, to April 1, 2024. It consists of monthly observations, where each monthly value represents the average of the available daily data. These series  are not seasonally adjusted. This characteristic makes them well-suited for analyzing long-term trends, volatility, and potential structural changes in the foreign exchange market.

Covering a 25-year period, the data set captures key economic events that have shaped currency fluctuations, including the introduction of the Euro, the 2008 financial crisis, Brexit, and the COVID-19 pandemic. To conduct the analysis, negative log returns were computed as
\begin{equation*}
    r_t = \log\left(\frac{p_{t-1}}{p_{t}} \right),
\end{equation*}
where \( p_t \) and \( p_{t-1} \) represent the exchange rate at month \( t \) and \( t-1 \), respectively.

\subsection{The model}
Empirical research on multivariate time series has shown that asset returns often exhibit time-varying dependence (see \citet{patton2001modelling}). To capture this dynamic behavior, we adopt the copula time series model proposed by \citet{patton2006modelling}, which he applied to analyze the dependence between the Deutsche Mark and Japanese Yen exchange rates. This approach preserves a fixed copula functional form throughout the sample while allowing its parameters to evolve according to a specified equation. Estimating these parameters requires an initial step of defining the marginal distributions for asset log returns. The model has been widely used in studies of asset return comovements under time-varying dependence frameworks (see \citet{patton2006estimation},  \citet{reboredo2011crude,reboredo2013gold}, and \citet{ji2019information}). Alternative approaches to modeling dynamic dependence have also been proposed, including semiparametric methods for conditional copula estimation  (\citet{acar2011conditional}; \citet{Abegaz2012}) and frameworks based on non-stationary random vectors, where both the marginal distributions and the copula may evolve over time (\citet{Nasri2019}).

We define the negative log-returns of the US dollar to UK pound sterling spot exchange rate as  \(X_t\) and the US dollar to euro spot exchange rate as \(Y_t\). Numerous studies have shown that exchange rate returns, typically measured as log changes in the spot exchange rate, exhibit autocorrelation, volatility clustering, and conditional heteroskedasticity (see  \citet{mcguirk1993modeling}  and references therein). To account for these properties, we model their marginal conditional distributions using ARMA-GARCH specifications.  The analyses presented in Section \ref{data} support the proposed models for the two marginal series. The marginal distribution of $X_t$ is specified as an AR(1)+GARCH(1,1) model with iid Student \textit{t} innovations:
\begin{eqnarray*}
X_t &=&  \phi_1 X_{t-1} + \varepsilon_t, \quad \varepsilon_t = \sigma_{x,t} a_t, \quad a_t \sim \ t_{m_1} \\
 \sigma_{x,t}^2&=&\omega_x+\alpha_x\varepsilon_{t-1}^2+\beta_x\sigma_{x,t-1}^2
\end{eqnarray*}
Similarly, the marginal distribution of \( Y_t \) is modeled as an MA(1) + GARCH(1,1)  process with iid skew Student-\textit{t} distributed innovations, given by:
\begin{eqnarray*}
	Y_t&=& \theta_1 \eta_{t-1} +\eta_t  \quad \eta_t=\sigma_{y,t}b_t, \quad  b_t\sim   \ t^*_{m_2, \xi} \\ \quad \sigma_{y,t}^2 &=& \omega_{y} + \alpha_{y} \eta_{t-1}^2 + \beta_{y} \sigma_{y,t-1}^2 
\end{eqnarray*}
where \( t^*_{m_2,\xi} \) denotes a standardized skew Student-\textit{t} distribution with skewness parameter \( \xi \) and degrees of freedom \( m_2 \). Its density function is given by:
\[
g_{m_2,\xi}\left(x \right) =
\begin{cases}
\frac{2}{\xi + \frac{1}{\xi}} \, f_{m_2}\left(\xi x \right), & \text{if } x < 0, \\
\frac{2}{\xi + \frac{1}{\xi}} \, f_{m_2}\left(\frac{x}{\xi}\right), & \text{if } x \geq 0,
\end{cases}
\]
where $f_{m_2}$ is the density function of a Student-\textit{t} distribution with $m_2$ degrees of freedom and  \( \xi \) is the skewness parameter.

Denote by $\boldsymbol{H}_t(\cdot;\boldsymbol{\lambda}_c)$ the conditional distribution function of the bivariate time serie $\boldsymbol{X}_t=(X_t,Y_t),$ given the information set at time $t-1.$
By applying  Sklar's theorem to the joint conditional distribution function, we have	
\begin{equation}\label{jointdf}
	\boldsymbol{H}_t(x, y ; \boldsymbol{\lambda})=C_t\left(F_{t}\left(x ; \lambda_1 \right), G_{t}\left(y ; \lambda_2 \right) ; \lambda_c\right), 
\end{equation}
where $\lambda_1$ and $\lambda_2$ are the parameters for the marginal conditional distributions, $\lambda_c$ are the parameters for the conditional copula and $\boldsymbol{\lambda}=(\lambda_1,\lambda_2,\lambda_c)$ are the parameters for the joint conditional distribution. Although our model incorporates time-varying copulas, we initially considered a static (or time-invariant) Student-\textit{t} copula. This model was fitted to the data using the \texttt{fitCopula} function from the \textit{copula} package in \textsf{R}, which implements a semi-parametric maximum pseudo-likelihood estimator based on pseudo-observations. The estimation yielded a copula parameter of $n = 9.7595$ degrees of freedom (we assume that the degrees of freedom parameter is constant and that only the correlation parameter is time-varying). The goodness-of-fit test for bivariate copulas, based on White's information matrix equality (\citet{white1982maximum}), yielded a $p$-value of 0.88 for the Student-\textit{t} copula. This result suggests that the Student-\textit{t} copula is a suitable choice, as there is insufficient evidence to reject the model at conventional significance levels under the assumption of time-invariant dependence.

Following the work of  \citet{patton2006estimation}, we model the  dependence parameter \(\rho_t\) for the Student-\textit{t} copula using an ARMA(1,10)-type process:  
\begin{equation}\label{rot}
    \rho_t = \Lambda_1\left(\nu_0 + \nu_1 \rho_{t-1} + \nu_2 \frac{1}{10} \sum_{j=1}^{10}  t_n^{-1}(u_{t-j}) \cdot t_n^{-1}(v_{t-j})\right)
\end{equation}
where \(t_n^{-1}\) is the inverse cumulative distribution function of the \(t\)-distribution with \(n\) degrees of freedom, and \(\Lambda_1(x) = \frac{1 - e^{-x}}{1 + e^{-x}}\) is a modified logistic transformation that ensures \(\rho_t\) remains within \((-1,1)\). The parameter vector \(\lambda_c = (\nu_0, \nu_1, \nu_2)\) governs the evolution of \(\rho_t\).  
This model accounts for the persistence of dependence by including \(\rho_{t-1}\) as a regressor and incorporates the intuition that correlation strengthens when the transformed marginals share the same sign and weakens when they have opposite signs.

\subsection{Estimation and Testing}

The set of copula parameters, \(\lambda_c\), is estimated using the maximum likelihood method. Given a random sample \((x_t, y_t)_{t=1}^{n_0}\), the log-likelihood function, following equation \eqref{jointdf}, is expressed as:
\begin{equation*}
    l(\boldsymbol{\lambda}) = \sum_{t=1}^{n_0} \left\{ \log f_t(x_t ; \lambda_1) + \log g_t(y_t ; \lambda_2) + \log c_t\left(F_t(x_t ; \lambda_1), G_t(y_t ; \lambda_2); \lambda_c\right) \right\},
\end{equation*}
where \( f_t \) and \( g_t \) denote the marginal conditional density functions, and \( c_t \) represents the conditional copula density function.

To estimate the parameters, we adopt a two-stage maximum likelihood procedure, as proposed by \citet{joe1996estimation}. In the first stage, we estimate the parameters of the marginal distributions independently. In the second stage, we estimate the dependence parameter by maximizing the copula likelihood, solving the following:
\begin{equation*}
    \hat{\lambda}_c = \underset{\lambda_c}{\arg\max} \sum_{t=1}^{n_0} \log c_t\left(\hat{u}_t, \hat{v}_t; \lambda_c \right),
\end{equation*}
where the pseudo-sample observations from the copula are given by \(\hat{u}_t = F_t(x_t ; \hat{\lambda}_1 )\) and \(\hat{v}_t = G_t(y_t ; \hat{\lambda}_2)\).  Under standard regularity conditions, this estimation procedure ensures the consistency and asymptotic normality of the estimates  (see \citet{joe1997multivariate})

\subsection{Data analysis}\label{data}
Table \ref{summary} presents the descriptive statistics for the log return series. The Shapiro-Wilk test strongly rejects the normality of the USD/GBP data, and both USD/GBP and USD/EUR exhibit positive excess kurtosis. The empirical correlation coefficient between the two series is 0.6822.

The parameters of the ARMA(p,q)-GARCH(r,s) models were empirically determined by selecting the optimal models from among the alternatives based on the Akaike Information Criterion (AIC). The Jarque-Bera test strongly rejected the normality of residuals for USD/GBP, but not for USD/EUR. To validate the volatility equation, we applied the Ljung-Box test to the squared standardized residuals. The sample autocorrelation function (ACF) and the p-values of the Kolmogorov-Smirnov and Ljung-Box tests indicate that the models are appropriate. Table \ref{marginals_table} presents the parameter estimates for the marginal models, while Figure \ref{scater} shows the scatterplot of the empirical copula, \( (F_{t}(x_t ; \hat{\lambda}_1), G_{t} (y_t ; \hat{\lambda}_2)) \).

\begin{table}[H]
	\begin{center}
		\begin{tabular}{| l | c | c |}	
			\hline		 & USD/EUR $(X_t)$&  USD/GBP $(Y_t$) \\
			 \hline
			 Mean & 0.000257 &0.000911 \\
			Std dev  & 0.022035 &  0.020911\\ 
			Max &0.077988	 &0.095447     \\
			Min &-0.061934 &-0.059854  \\
			Skewness& -0.000702&0.580742 \\
			Kurtosis & 0.468576&1.803208 \\
			Shapiro-Wilk $p$-val &0.2747& 4.473 $\times 10^{-5}$  \\

			Pearson's r & \multicolumn{2}{ |c|}{0.6822} \\
			number observed & \multicolumn{2}{ |c|}{303} \\
            
			\hline
		\end{tabular}
		\caption{Descriptive statistics for log returns. }
		\label{summary}
	\end{center}
\end{table}

	\begin{figure}[H]
	\centering
	\includegraphics[width=0.7\linewidth]{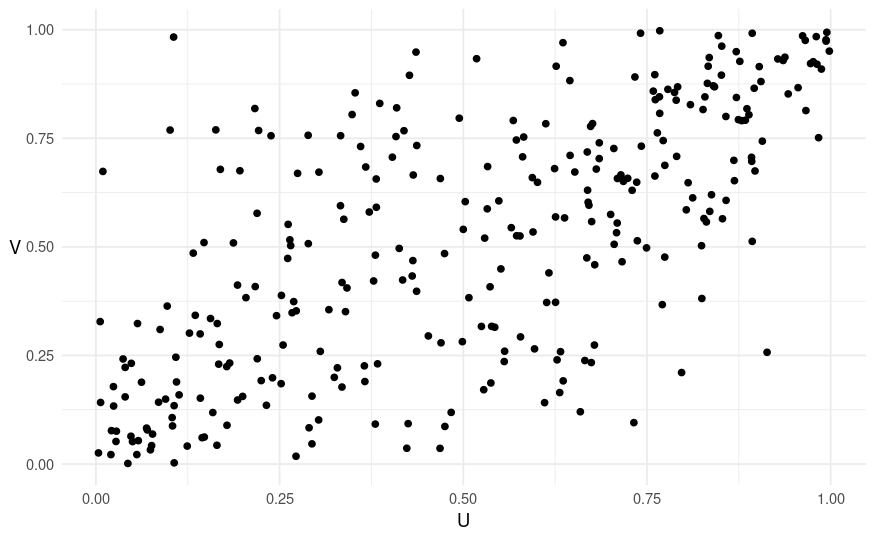}
	\caption{Scatter-plot of the (time-invariant) empirical copula of the bivariate time series $(X_t,Y_t),$ given by $(F_{t}(x_t ; \hat{\lambda}_1), G_{t} (y_t ; \hat{\lambda}_2)).$}
	\label{scater}
\end{figure}
\begin{table}[H]
	\begin{center}
\begin{tabular}{| l | c | l |c |}
			\hline
\multicolumn{4}{ |l| }{Mean equation} \\ 
			\hline
			\multicolumn{2}{ |l|} {USD/EUR $(X_t)$} &\multicolumn{2}{ |l|} {USD/GBP  $(Y_t)$} \\
\hline		
$\hat{\phi}_1$ & 0.0292 (0.0552) & $\hat{\theta}_1$& 0.2299 (0.0582) \\

\hline
\multicolumn{4}{ |l| }{Variance equation} \\ \hline
$\hat{\alpha}_{x}$	 &0.0675 (0.0369)& $\hat{\alpha}_y$ & 0.08609 (0.04883)\\
$\hat{\beta}_x$  & 0.9035 (0.0549)& $\hat{\beta}_y$& 0.7819 ( 0.1197)\\
Ljung-Box $R^2$ $p$-val&0.1143& & 0.7255 \\

\hline
\multicolumn{4}{ |l| }{Marginal models} \\
\hline
 $\hat{m}_1$& 10 	&  $\hat{m}_2$ & 10 \\
& &$\hat{\xi}$ & 1.298 \\ 
Jarque-Bera $p$-val& 0.1978& &	0.5 $\times 10^{-5}$	 \\
		Ljung-Box (RS) $p$-val & 0.6099	 & &0.9498\\
		\hline
\end{tabular}
		\caption{Maximum likelihood estimates with asymptotic standard errors in parentheses of the parameters of the marginal distribution models for EXUSEU and EXUSUK log returns. Ljung–Box test for the squared residuals $(R^2)$ and for standardized residuals (RS) are computed with 20 lags. Jarque-Bera tests the normality of residuals. }
		\label{marginals_table}
	\end{center}
\end{table}

\subsection{Results}
In this section, we analyze the bivariate time series \(\{(X_t, Y_t)\}\), where \(X_t\) denotes the negative log-returns of the USD/EUR spot exchange rate, and \(Y_t\) corresponds to those of the USD/GBP rate. Our objective is to monitor the risk exposure associated with the USD/GBP exchange rate (\(Y_t\)) by observing the auxiliary series USD/EUR (\(X_t\)) and analyzing the associated \(\text{PELCoV}_v\).

\subsection{Existence, uniqueness and calculation of \text{PELCoV}$_v$}

According to Theorem~\ref{main_theorem}, for any \( v > \frac{1}{2} \), the existence of a \text{PELCoV}$_v$ greater than \( \frac{1}{2} \) is guaranteed. However, the existence of a second \text{PELCoV}$_v$ depends on whether the condition \( v > L_0 \) is satisfied.

It is important to recall that, in our time-varying copula model, the degrees of freedom of the $t$-distribution are fixed, while the dependence parameter \( \rho \) evolves dynamically according to Equation~\ref{rot}. Since \( L_0 \) is a function of \( \rho \), the existence of a second \text{PELCoV}$_v$ at time \( t \) is governed by the inequality \( v > L_0^t \), where
\begin{equation}
  L_0^t = t_{n+1} \left( \frac{\rho_t \sqrt{n+1}}{\sqrt{1 - \rho_t^2}} \right),
\end{equation}
and \( \rho_t \) denotes the value of \( \rho \) at time \( t \).

Figure~\ref{L0g} illustrates the evolution of \( L_0^t \) as a function of \( \rho_t \) under a Student-\textit{t} distribution with 9.7595 degrees of freedom. In our dataset, the minimum value of \( L_0^t \) across all time points satisfies
\[
\min_t \{ L_0^t \} > 0.9948,
\]
which exceeds the conventional probability levels typically adopted for risk control.

In this study, we consider \( v = 0.99 \) and \( v = 0.95 \), both of which fall below this threshold. Consequently, only a single \text{PELCoV}$_v$ exists at each point in time. According to Lemma~\ref{ttkm}, given \( v \in (0,1) \), the \text{PELCoV}$_v$ at time \( t \) is the valid solution \( u_v \) to the equation
\begin{equation*}
u_v=t_{n}\left(\frac{t_{n}^{-1}(v)\rho_t(n+1)\pm \sqrt{k}}{\left(\rho_t^2(n+1)-\left(t_{n+1}^{-1}(v)\right)^2(1-\rho_t^2)\right)}\right)
\end{equation*}
where $$ k=\left((1-\rho_t^2)\left(t_{n+1}^{-1}(v)\right)^2\right)
\left( \rho_t^2(n+1)n^2+\left(t_{n}^{-1}(v)\right)^2(1+n\rho_t^2)\right)\ge 0.$$

\begin{figure}[h!]
    \centering
           \includegraphics[width=\linewidth]{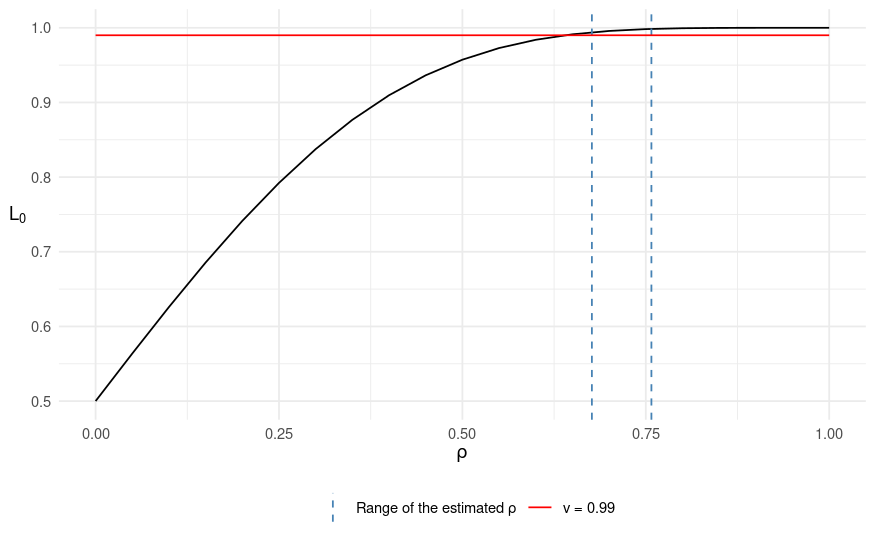}
        
        \caption{ Evolution of $L^t_0$ as a function of $\rho_t$. Within the range of $\rho_t$ values observed in our data (blue dashed lines), $L_0$ remains  above the threshold level $v = 0.99$ (red line) for all $t$. }\label{L0g}
   \end{figure}

\subsection{Interpretation of PELCoV$_v$}
The upper panel of Figure~\ref{pelcov99} displays the time series \( X_t \), along with the corresponding \( u_v(t) \)-quantiles, given by
\[
F_{X_t}^{-1}(u_v(t)),
\]
for \( v = 0.99 \), where \( u_v(t) \) denotes the \text{PELCoV}$_v$ at time \( t \). The lower panel presents the time series \( Y_t \), together with its Value at Risk at the 0.99 level, defined as
\[
\text{VaR}_{0.99}[Y_t] = G_{Y_t}^{-1}(0.99).
\]

\begin{figure}[H]
    \centering
    \begin{subfigure}[b]{\linewidth}
        \centering
        \includegraphics[width=\linewidth]{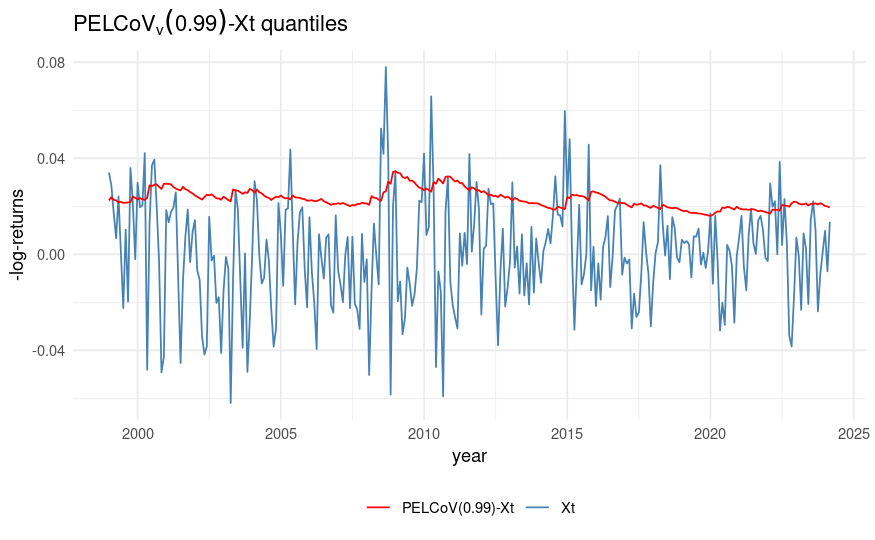}
    \end{subfigure}
    
    \vspace{0.3cm} 
    
    \begin{subfigure}[b]{\linewidth}
        \centering
        \includegraphics[width=\linewidth]{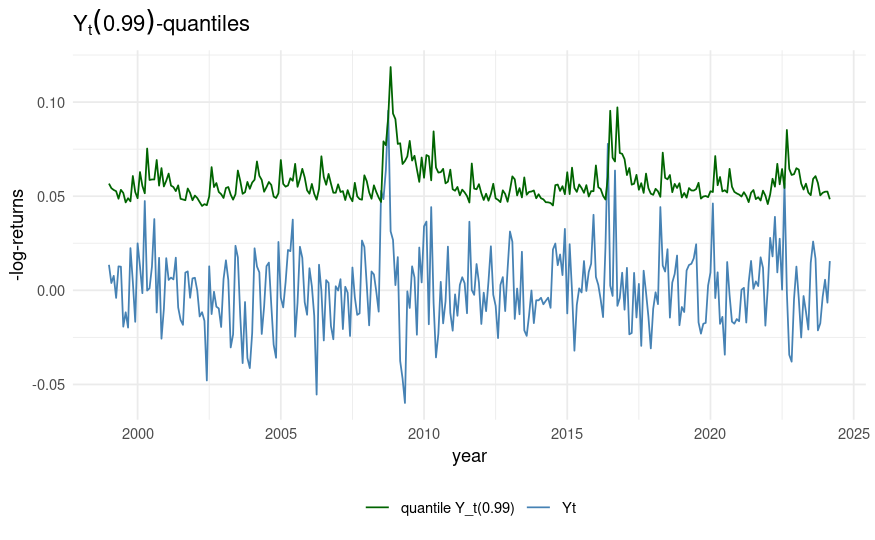}
    \end{subfigure}

    \caption{The upper panel shows the time series of the negative log-returns of the U.S. Dollar to Euro Spot Exchange Rate (EXUSEU), along with the \text{PELCoV}$_{0.99}$  under the assumption of a time-varying Student-\textit{t} copula. The lower panel displays the time series of the negative log-returns of the U.S. Dollar to British Pound (EXUSUK), together with the corresponding Value at Risk at the 0.99 level. }
    \label{pelcov99}
\end{figure}

\begin{figure}[H]
    \centering
    \begin{subfigure}[b]{\linewidth}
        \centering
        \includegraphics[width=\linewidth]{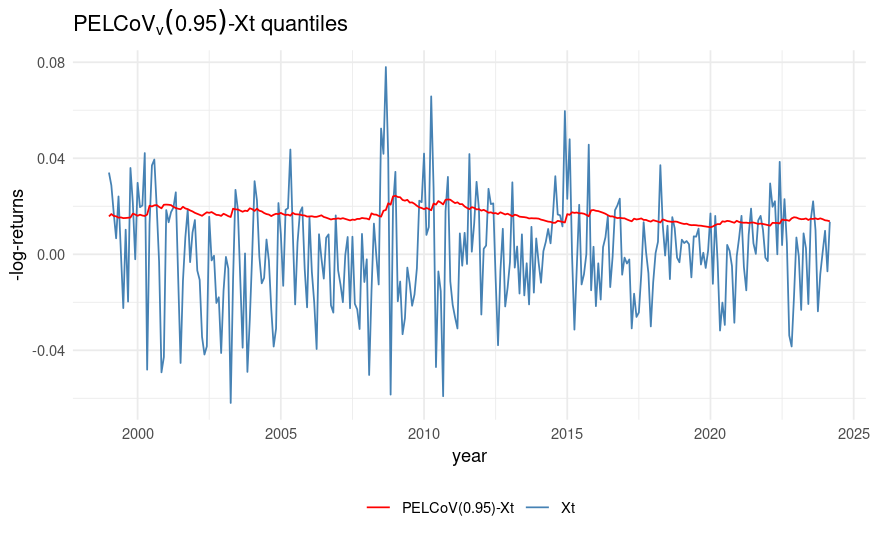}
    \end{subfigure}
    
    \vspace{0.3cm} 
    
    \begin{subfigure}[b]{\linewidth}
        \centering
        \includegraphics[width=\linewidth]{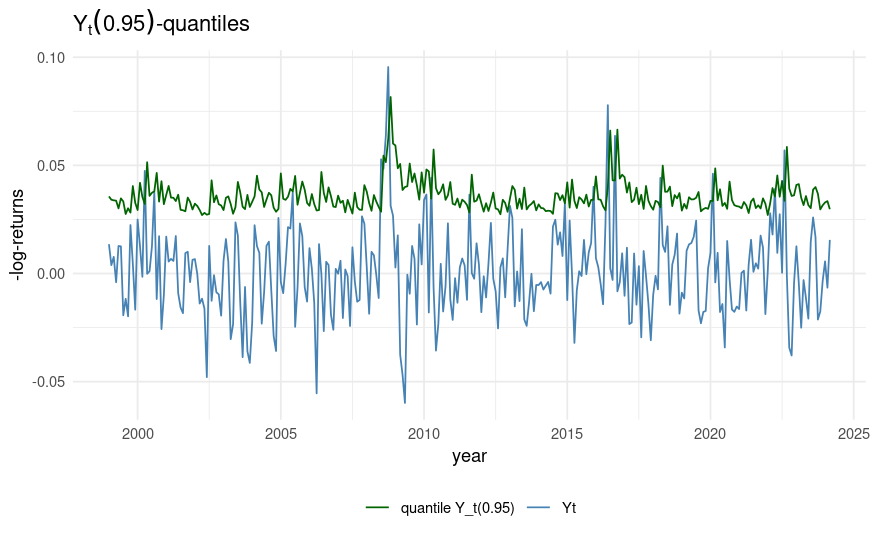}
    \end{subfigure}

\caption{The upper panel shows the time series of the negative log-returns of the U.S. Dollar to Euro Spot Exchange Rate (EXUSEU), along with the \text{PELCoV}$_{0.95}$  under the assumption of a time-varying Student-\textit{t} copula. The lower panel displays the time series of the negative log-returns of the U.S. Dollar to British Pound (EXUSUK), together with the corresponding Value at Risk at the 0.95 level.
}  
    \label{pelcov95}
\end{figure}

The upper panel illustrates, at each time point \( t \), the relationship between the observed values of \( X_t \) and the minimum loss threshold beyond which the risk associated with \( Y_t \) is underestimated when relying solely on \(\text{VaR}_v[Y_t]\) rather than the conditional risk measure \(\text{CoVaR}_{v,u}[Y_t \mid X_t]\).
Whenever the blue line (representing the values of \( X_t \)) exceeds the red line, which marks the values of \( X_t \) satisfying
\begin{equation}\label{eqvc}
  \text{VaR}_v[Y_t] = \text{CoVaR}_{v,u_v(t)}\left[Y_t \mid X_t\right],
\end{equation}
with \( v = 0.99 \), a prudent risk manager might consider replacing the Value at Risk with the Conditional Value at Risk as a more conservative measure. In fact, the model indicates that if \( X_t \) took any value \( F_{X_t}^{-1}(u') \) at time \( t \), with \( u' > u_v(t) \),  then it would follow that
\[
\text{VaR}_v[Y_t] < \text{CoVaR}_{v,u'}\left[Y_t \mid X_t\right],
\]
implying that \(\text{VaR}_v[Y_t]\) would underestimate the spillover risk originating from \( X_t \).

Figure~\ref{pelcov95} presents an analysis analogous to that of Figure~\ref{pelcov99} for a risk level of \( v = 0.95 \). Key insights from the analysis include the following:
\begin{enumerate}
    \item [(a)] A low-frequency event such as \(\text{VaR}_{0.99}[Y_t]\) offers limited utility as an early warning signal; by the time the series reaches this threshold, it is often too late to take preventive action against extreme risk. In contrast, the series \( X_t \) crosses its associated \( u_v(t) \) quantile more frequently, as \( u_v(t) \) fluctuates within the range \([0.678, 0.757]\). This makes \(\text{PELCoV}_v\) a more practical and informative tool for anticipating risk. More generally, leveraging a more stable and less volatile auxiliary variable than the one being monitored facilitates the early detection of extreme events at a manageable frequency.

   \item [(b)]  Observe that the upper panel of Figure~\ref{pelcov99} successfully anticipates several episodes of extreme risk during which the series \( Y_t \) reached its Value at Risk level \(\text{VaR}_{0.99}[Y_t]\). The most notable case occurs in 2022. Between February and June, the series \( X_t \) remained consistently above the \(\text{PELCoV}_{0.99}\) threshold, indicating that the Value at Risk was underestimating the risk associated with \( Y_t \). A risk manager who had adjusted their strategy in response to this early warning signal would not have been caught off guard by the extreme behavior exhibited by \( Y_t \) in the second quarter of 2022---particularly in June, when the series reached the level \(\text{VaR}_{0.99}[Y_t]\), as shown in the lower panel of Figure~\ref{pelcov99}. This episode coincided with heightened financial volatility driven by aggressive interest rate hikes by the Federal Reserve and mounting concerns over a potential global recession.

\item[(c)] Another noteworthy episode occurred in early 2000, during the period of heightened financial market volatility associated with the burst of the dot-com bubble. As shown in the upper panel of Figure~\ref{pelcov99}, between January and April the series \( X_t \) remained persistently above the \(\text{PELCoV}_{0.99}\) threshold, indicating a regime in which the Value-at-Risk measure was likely underestimating the risk associated with \( Y_t \). In April 2000, the series \( Y_t \) approached values close to its \(\text{VaR}_{0.99}[Y_t]\), validating the early warning signaled by the behavior of the auxiliary variable. This is most clearly illustrated by comparing the upper and lower panels of Figure~\ref{pelcov95}, where the \(\text{PELCoV}_{0.95}\) threshold effectively anticipates the extreme values reached by the series \( Y_t \), which even touched its \(\text{VaR}_{0.95}[Y_t]\).

A similar pattern emerged in late 2015. The series \( X_t \) remained above the \(\text{PELCoV}_{0.95}\) threshold for several months, conveying the signal that the \(\text{VaR}_{0.95}[Y_t]\) was underestimating risk relative to the CoVaR. Eventually, in December of that year, \( Y_t \) reached the level \(\text{VaR}_{0.95}[Y_t]\). This period coincided with heightened currency market volatility driven by the Federal Reserve’s first interest rate hike in nearly a decade, the European Central Bank’s expansion of its quantitative easing program, and lingering uncertainty following the Greek debt crisis.

\end{enumerate}

These episodes underscore the potential of \(\text{PELCoV}_v\) to detect early signs of risk escalation, even during periods of moderate volatility.

\section{Conclusions}
Building on the recent development of the probability-equivalent level of VaR and CoVaR (PELCoV) framework introduced by Ortega-Jiménez et al. (2024), this paper extends the methodology to accommodate bivariate risks modeled by a Student-\textit{t} copula, thereby relaxing the strong dependence assumptions of earlier work. Although the theoretical results are established in a static setting, we implement them dynamically to capture evolving dependence structures over time.

Our empirical application focuses on the foreign exchange market, monitoring the USD/GBP exchange rate with the USD/EUR series serving as an auxiliary variable. Covering the period from 1999 to 2024, the results show that \(\text{PELCoV}_v\) is a valuable tool for early warning detection of extreme risk episodes, even under moderate market volatility.

Notably, the methodology successfully anticipated key stress periods such as the dot-com bubble in early 2000, the financial turbulence in late 2015, and the volatility spike associated with monetary policy tightening in 2022. These findings underscore the practical utility of extending the PELCoV approach to more flexible dependence structures, enhancing its relevance for risk monitoring in financial markets.

\section{Acknowledgement}

MAS and ASL acknowledge financial support from the Spanish Ministry of Economy and Competitiveness through grant PID2020-116216GB-I00. DIFS acknowledges the University of Cádiz for the PhD scholarship linked to the same project (PID2020-116216GB-I00), funded by the Spanish Ministry of Economy and Competitiveness.

\section*{Appendix A}
To analyze the critical points of $h(u),$ we examine its derivative
\[\frac{\partial}{\partial u}\left(\frac{t_n^{-1}(v)-\rho t_n^{-1}(u)}{\sqrt{n+(t_n^{-1}(u))^2}} \right).\]
Expanding this expression leads to:
\[
\frac{\left(\frac{-\rho}{dt_n\left( t_n^{-1}(u)\right)} \right)\left( \sqrt{n+(t_n^{-1}(u))^2}\right)-\left( t_n^{-1}(v)-\rho t_n^{-1}(u)\right)\left(\frac{2  t_n^{-1}(u)\left(n+(t_n^{-1}(u))^2\right)^{-\frac{1}{2}}}{dt_n\left( t_n^{-1}(u)\right)} \right)}{n+(t_n^{-1}(u))^2}.
\]
Here, $dt_n$ denotes the probability density function of the Student \textit{t}-distribution with $n$ degrees of freedom. We proceed by simplifying the expression further.
\begin{align*}
    \frac{\partial}{\partial u}\left(\frac{t_n^{-1}(v)-\rho t_n^{-1}(u)}{\sqrt{n+(t_n^{-1}(u))^2}} \right) &= \frac{-\rho\left(n+(t_n^{-1}(u))^2\right)-t_n^{-1}(u)\left(t_n^{-1}(v)-\rho t_n^{-1}(u)\right)}{dt_n\left(t_n^{-1}(u)\right)\left(n+(t_n^{-1}(u))^2\right)^{\frac{3}{2}}}\\
    &= \frac{-\rho n  -\rho (t_n^{-1}(u))^2  -t_n^{-1}(u)t_n^{-1}(v) +\rho (t_n^{-1}(u))^2}{dt_n\left(t_n^{-1}(u)\right)\left(n+(t_n^{-1}(u))^2\right)^{\frac{3}{2}}}\\
    &= \frac{-\rho n-t_n^{-1}(u)t_n^{-1}(v)}{dt_n(t_n^{-1}(u))\left(n+(t_n^{-1}(u))^2\right)^{3/2}}.
\end{align*}
Therefore, $\frac{\partial}{\partial u}h(u)=0$ if and only if $-\rho n-t_n^{-1}(u)t_n^{-1}(v)=0$.  The only point satisfying this condition is given by $u^*=t_n\left(\frac{-\rho n}{t_n^{-1}(v)}\right)$.

\section*{Appendix B}
In this appendix, we provide the proof of Lemma \ref{ttkm}. Given a value \( v \in (0,1) \), the PELCoV\(_v\) is defined as the solution \( u_v \) to the equation \( h(u_v) = v \), where the function \( h(u) \) is defined in Equation~\eqref{hu}. By using the properties of the inverse cumulative distribution function \( t_{n+1}^{-1} \), we can express this equation equivalently as:
\begin{equation}
\frac{t_n^{-1}(v) - \rho\, t_n^{-1}(u_v)}{\sqrt{\dfrac{n + (t_n^{-1}(u_v))^2}{n+1}(1 - \rho^2)}} = t_{n+1}^{-1}(v).
\end{equation}

To simplify the expression, let us denote:
\[
a = t_n^{-1}(u_v), \quad b = t_n^{-1}(v), \quad c = t_{n+1}^{-1}(v).
\]
Substituting these into the equation, we obtain:
\begin{equation*}
\frac{(b - \rho a)^2}{\dfrac{(n + a^2)(1 - \rho^2)}{n+1}} = c^2.
\end{equation*}
Multiplying both sides and rearranging terms yields the following quadratic equation in \( a \):
\[
a^2\left( \rho^2(n+1) - c^2(1 - \rho^2) \right) - 2ab\rho(n+1) + b^2(n+1) - nc^2(1 - \rho^2) = 0.
\]
This quadratic equation has solutions given by:
\[
a = \frac{b\rho(n+1) \pm \sqrt{k}}{\rho^2(n+1) - c^2(1 - \rho^2)},
\]
where the discriminant \( k \) is defined as:
\[
k = b^2 \rho^2 (n+1)^2 - \left( \rho^2(n+1) - c^2(1 - \rho^2) \right) \left( b^2(n+1) - nc^2(1 - \rho^2) \right).
\]

Returning to the original variables, we obtain:
\[
t_n^{-1}(u_v) = \frac{t_n^{-1}(v)\rho(n+1) \pm \sqrt{k}}{\rho^2(n+1) - \left( t_{n+1}^{-1}(v) \right)^2 (1 - \rho^2)}.
\]
Finally, solving for \( u_v \), we arrive at the expression:
\[
u_v = t_n\left( \frac{t_n^{-1}(v)\rho(n+1) \pm \sqrt{k}}{\rho^2(n+1) - \left( t_{n+1}^{-1}(v) \right)^2 (1 - \rho^2)} \right).
\]

\bibliographystyle{chicago}
\bibliography{bibliography/bibliography}

\begin{thebibliography}{}

\bibitem[\protect\citeauthoryear{Aas, Czado, Frigessi, and Bakken}{Aas
  et~al.}{2009}]{aas2009pair}
Aas, K., C.~Czado, A.~Frigessi, and H.~Bakken (2009).
\newblock Pair-copula constructions of multiple dependence.
\newblock {\em Insurance Math. Econom.\/}~{\em 44\/}(2), 182--198.

\bibitem[\protect\citeauthoryear{Abegaz, Gijbels, and Veraverbeke}{Abegaz
  et~al.}{2012}]{Abegaz2012}
Abegaz, F., I.~Gijbels, and N.~Veraverbeke (2012).
\newblock Semiparametric estimation of conditional copulas.
\newblock {\em Journal of Multivariate Analysis\/}~{\em 110}, 43--73.

\bibitem[\protect\citeauthoryear{Acar, Craiu, and Yao}{Acar
  et~al.}{2011}]{acar2011conditional}
Acar, E.~F., R.~V. Craiu, and F.~Yao (2011).
\newblock Dependence calibration in conditional copulas: A nonparametric
  approach.
\newblock {\em Biometrics\/}~{\em 67\/}(2), 445--453.

\bibitem[\protect\citeauthoryear{Acharya, Pedersen, Philippon, and
  Richardson}{Acharya et~al.}{2017}]{acharya2017measuring}
Acharya, V.~V., L.~H. Pedersen, T.~Philippon, and M.~Richardson (2017).
\newblock Measuring systemic risk.
\newblock {\em Rev. Financ. Stud.\/}~{\em 30\/}(1), 2--47.

\bibitem[\protect\citeauthoryear{Albulescu and Pepin}{Albulescu and
  Pepin}{2018}]{albulescu2018monetary}
Albulescu, C.~T. and D.~Pepin (2018).
\newblock Monetary integration, money-demand stability, and the role of
  monetary overhang in forecasting inflation in cee countries.
\newblock {\em J. Econ. Integr.\/}~{\em 33\/}(4), 841--879.

\bibitem[\protect\citeauthoryear{Arias-Nicol{\'a}s, Fern{\'a}ndez-Ponce,
  Luque-Calvo, and Su{\'a}rez-Llorens}{Arias-Nicol{\'a}s
  et~al.}{2005}]{arias2005multivariate}
Arias-Nicol{\'a}s, J., J.~Fern{\'a}ndez-Ponce, P.~Luque-Calvo, and
  A.~Su{\'a}rez-Llorens (2005).
\newblock Multivariate dispersion order and the notion of copula applied to the
  multivariate t-distribution.
\newblock {\em Probab. Eng. Inform. Sci.\/}~{\em 19\/}(3), 363--375.

\bibitem[\protect\citeauthoryear{Benoit, Colliard, Hurlin, and
  P{\'e}rignon}{Benoit et~al.}{2017}]{benoit2017risks}
Benoit, S., J.-E. Colliard, C.~Hurlin, and C.~P{\'e}rignon (2017).
\newblock Where the risks lie: A survey on systemic risk.
\newblock {\em Rev. Finance\/}~{\em 21\/}(1), 109--152.

\bibitem[\protect\citeauthoryear{Beutner, Heinemann, and Smeekes}{Beutner
  et~al.}{2024}]{beutner2024residual}
Beutner, E., A.~Heinemann, and S.~Smeekes (2024).
\newblock A residual bootstrap for conditional value-at-risk.
\newblock {\em J. Econometrics\/}~{\em 238\/}(2), 105554.

\bibitem[\protect\citeauthoryear{Bisias, Flood, Lo, and Valavanis}{Bisias
  et~al.}{2012}]{bisias2012survey}
Bisias, D., M.~Flood, A.~W. Lo, and S.~Valavanis (2012).
\newblock A survey of systemic risk analytics.
\newblock {\em Annu. Rev. Financ. Econ.\/}~{\em 4\/}(1), 255--296.

\bibitem[\protect\citeauthoryear{Cai and Wei}{Cai and
  Wei}{2012}]{cai2012invariant}
Cai, J. and W.~Wei (2012).
\newblock On the invariant properties of notions of positive dependence and
  copulas under increasing transformations.
\newblock {\em Insurance Math. Econom.\/}~{\em 50\/}(1), 43--49.

\bibitem[\protect\citeauthoryear{Chen, Iyengar, and Moallemi}{Chen
  et~al.}{2014}]{chen2014asset}
Chen, C., G.~Iyengar, and C.~C. Moallemi (2014).
\newblock Asset-based contagion models for systemic risk.
\newblock {\em Columbia Business School Working Paper\/}.

\bibitem[\protect\citeauthoryear{Chen, Fan, and Patton}{Chen
  et~al.}{2004}]{chen2004simple}
Chen, X., Y.~Fan, and A.~J. Patton (2004).
\newblock Simple tests for models of dependence between multiple financial time
  series, with applications to us equity returns and exchange rates.
\newblock {\em London Economics Financial Markets Group Working Paper No.
  483\/}.

\bibitem[\protect\citeauthoryear{Diks, Panchenko, and Van~Dijk}{Diks
  et~al.}{2010}]{diks2010out}
Diks, C., V.~Panchenko, and D.~Van~Dijk (2010).
\newblock Out-of-sample comparison of copula specifications in multivariate
  density forecasts.
\newblock {\em J. Econ. Dyn. Control\/}~{\em 34\/}(9), 1596--1609.

\bibitem[\protect\citeauthoryear{Du and Lai}{Du and Lai}{2017}]{du2017copula}
Du, J. and K.~K. Lai (2017).
\newblock Copula-based risk management models for multivariable rmb exchange
  rate in the process of rmb internationalization.
\newblock {\em J. Syst. Sci. Complex.\/}~{\em 30}, 660--679.

\bibitem[\protect\citeauthoryear{Filipiak, Klein, Mazur, and
  Mrowińska}{Filipiak et~al.}{2025}]{FILIPIAK2025105490}
Filipiak, K., D.~Klein, S.~Mazur, and M.~Mrowińska (2025).
\newblock Likelihood ratio test for covariance matrix under multivariate t
  distribution with uncorrelated observations.
\newblock {\em J. Multivar. Anal.\/}~{\em 210}, 105490.

\bibitem[\protect\citeauthoryear{Francq and Zako{\"i}an}{Francq and
  Zako{\"i}an}{2025}]{francq2025inference}
Francq, C. and J.-M. Zako{\"i}an (2025).
\newblock Inference on dynamic systemic risk measures.
\newblock {\em J. Econometrics\/}~{\em 247}, 105936.

\bibitem[\protect\citeauthoryear{Girardi and Erg{\"u}n}{Girardi and
  Erg{\"u}n}{2013}]{girardi2013systemic}
Girardi, G. and A.~T. Erg{\"u}n (2013).
\newblock Systemic risk measurement: Multivariate garch estimation of covar.
\newblock {\em J. Bank. Finance\/}~{\em 37\/}(8), 3169--3180.

\bibitem[\protect\citeauthoryear{Glasserman and Young}{Glasserman and
  Young}{2016}]{glasserman2016contagion}
Glasserman, P. and H.~P. Young (2016).
\newblock Contagion in financial networks.
\newblock {\em J. Econ. Lit.\/}~{\em 54\/}(3), 779--831.

\bibitem[\protect\citeauthoryear{Gong and Huser}{Gong and
  Huser}{2022}]{gong2022asymmetric}
Gong, Y. and R.~Huser (2022).
\newblock Asymmetric tail dependence modeling, with application to
  cryptocurrency market data.
\newblock {\em Ann. Appl. Stat.\/}~{\em 16\/}(3), 1822--1847.

\bibitem[\protect\citeauthoryear{Ji, Bouri, Roubaud, and Kristoufek}{Ji
  et~al.}{2019}]{ji2019information}
Ji, Q., E.~Bouri, D.~Roubaud, and L.~Kristoufek (2019).
\newblock Information interdependence among energy, cryptocurrency and major
  commodity markets.
\newblock {\em Energy Econ.\/}~{\em 81}, 1042--1055.

\bibitem[\protect\citeauthoryear{Joe}{Joe}{1997}]{joe1997multivariate}
Joe, H. (1997).
\newblock {\em Multivariate Models and Multivariate Dependence Concepts}.
\newblock CRC Press.

\bibitem[\protect\citeauthoryear{Joe}{Joe}{2014}]{joe2014dependence}
Joe, H. (2014).
\newblock {\em Dependence Modeling with Copulas}.
\newblock CRC Press.

\bibitem[\protect\citeauthoryear{Joe and Xu}{Joe and
  Xu}{1996}]{joe1996estimation}
Joe, H. and J.~J. Xu (1996).
\newblock The estimation method of inference functions for margins for
  multivariate models.
\newblock {\em University of British Columbia Working Paper\/}.

\bibitem[\protect\citeauthoryear{Jorion}{Jorion}{2000}]{jorion2000risk}
Jorion, P. (2000).
\newblock Risk management lessons from long-term capital management.
\newblock {\em Eur. Financ. Manag.\/}~{\em 6\/}(3), 277--300.

\bibitem[\protect\citeauthoryear{Lehmann}{Lehmann}{1966}]{lehmann1966theorem}
Lehmann, E. (1966).
\newblock On a theorem of bahadur and goodman.
\newblock {\em Ann. Math. Statist.\/}~{\em 37\/}(1), 1--6.

\bibitem[\protect\citeauthoryear{Liu, Semeyutin, Lau, and Gozgor}{Liu
  et~al.}{2020}]{liu2020forecasting}
Liu, W., A.~Semeyutin, C.~K.~M. Lau, and G.~Gozgor (2020).
\newblock Forecasting value-at-risk of cryptocurrencies with riskmetrics type
  models.
\newblock {\em Res. Int. Bus. Finance\/}~{\em 54}, 101259.

\bibitem[\protect\citeauthoryear{Loaiza-Maya, G{\'o}mez-Gonz{\'a}lez, and
  Melo-Velandia}{Loaiza-Maya et~al.}{2015}]{loaiza2015exchange}
Loaiza-Maya, R.~A., J.~E. G{\'o}mez-Gonz{\'a}lez, and L.~F. Melo-Velandia
  (2015).
\newblock Exchange rate contagion in latin america.
\newblock {\em Res. Int. Bus. Finance\/}~{\em 34}, 355--367.

\bibitem[\protect\citeauthoryear{Mainik and Schaanning}{Mainik and
  Schaanning}{2014}]{mainik2014dependence}
Mainik, G. and E.~Schaanning (2014).
\newblock On dependence consistency of covar and some other systemic risk
  measures.
\newblock {\em Statist. Risk Model.\/}~{\em 31\/}(1), 49--77.

\bibitem[\protect\citeauthoryear{McGuirk, Robertson, and Spanos}{McGuirk
  et~al.}{1993}]{mcguirk1993modeling}
McGuirk, A., J.~Robertson, and A.~Spanos (1993).
\newblock Modeling exchange rate dynamics: Non-linear dependence and thick
  tails.
\newblock {\em Econometric Rev.\/}~{\em 12\/}(1), 33--63.

\bibitem[\protect\citeauthoryear{M{\"u}ller and Stoyan}{M{\"u}ller and
  Stoyan}{2002}]{muller2002comparison}
M{\"u}ller, A. and D.~Stoyan (2002).
\newblock {\em Comparison Methods for Stochastic Models and Risks}.
\newblock Wiley.

\bibitem[\protect\citeauthoryear{Nasri, Rémillard, and Bouezmarni}{Nasri
  et~al.}{2019}]{Nasri2019}
Nasri, B.~R., B.~N. Rémillard, and T.~Bouezmarni (2019).
\newblock Semi-parametric copula-based models under non-stationarity.
\newblock {\em J. Multivar. Anal.\/}~{\em 173}, 347--265.

\bibitem[\protect\citeauthoryear{Nelsen}{Nelsen}{2005}]{nelsen2005copulas}
Nelsen, R.~B. (2005).
\newblock Copulas and quasi-copulas: An introduction to their properties and
  applications.
\newblock In {\em Logical, Algebraic, Analytic and Probabilistic Aspects of
  Triangular Norms}, pp.\  391--413. Elsevier.

\bibitem[\protect\citeauthoryear{Nelsen}{Nelsen}{2006}]{nelsen2006introduction}
Nelsen, R.~B. (2006).
\newblock {\em An Introduction to Copulas}.
\newblock Springer.

\bibitem[\protect\citeauthoryear{Ortega-Jim{\'e}nez, Pellerey, Sordo, and
  Su{\'a}rez-Llorens}{Ortega-Jim{\'e}nez et~al.}{2024}]{ortega2024}
Ortega-Jim{\'e}nez, P., F.~Pellerey, M.~A. Sordo, and A.~Su{\'a}rez-Llorens
  (2024).
\newblock Probability equivalent level for covar and var.
\newblock {\em Insurance Math. Econom.\/}~{\em 115}, 22--35.

\bibitem[\protect\citeauthoryear{Ortega-Jim{\'e}nez, Sordo, and
  Su{\'a}rez-Llorens}{Ortega-Jim{\'e}nez et~al.}{2021}]{ortega2021stochastic}
Ortega-Jim{\'e}nez, P., M.~A. Sordo, and A.~Su{\'a}rez-Llorens (2021).
\newblock Stochastic orders and multivariate measures of risk contagion.
\newblock {\em Insurance Math. Econom.\/}~{\em 96}, 199--207.

\bibitem[\protect\citeauthoryear{Patton}{Patton}{2001}]{patton2001modelling}
Patton, A.~J. (2001).
\newblock Modelling time-varying exchange rate dependence using the conditional
  copula.
\newblock {\em UC San Diego: Department of Economics, UCSD. Retrieved from
  https://escholarship.org/uc/item/01q7j1s2\/}.

\bibitem[\protect\citeauthoryear{Patton}{Patton}{2006a}]{patton2006estimation}
Patton, A.~J. (2006a).
\newblock Estimation of multivariate models for time series of possibly
  different lengths.
\newblock {\em J. Appl. Econom.\/}~{\em 21\/}(2), 147--173.

\bibitem[\protect\citeauthoryear{Patton}{Patton}{2006b}]{patton2006modelling}
Patton, A.~J. (2006b).
\newblock Modelling asymmetric exchange rate dependence.
\newblock {\em Int. Econ. Rev.\/}~{\em 47\/}(2), 527--556.

\bibitem[\protect\citeauthoryear{Reboredo}{Reboredo}{2011}]{reboredo2011crude}
Reboredo, J.~C. (2011).
\newblock How do crude oil prices co-move? a copula approach.
\newblock {\em Energy Econ.\/}~{\em 33\/}(5), 948--955.

\bibitem[\protect\citeauthoryear{Reboredo}{Reboredo}{2013}]{reboredo2013gold}
Reboredo, J.~C. (2013).
\newblock Is gold a safe haven or a hedge for the us dollar? implications for
  risk management.
\newblock {\em J. Bank. Finance\/}~{\em 37\/}(8), 2665--2676.

\bibitem[\protect\citeauthoryear{Shim and Lee}{Shim and
  Lee}{2017}]{ShimLee2017}
Shim, J. and S.-H. Lee (2017).
\newblock Dependency between risks and the insurer’s economic capital: A
  copula-based garch model.
\newblock {\em Asia-Pacific Journal of Risk and Insurance\/}~{\em 11\/}(1),
  1--29.

\bibitem[\protect\citeauthoryear{Shyamalkumar and Tao}{Shyamalkumar and
  Tao}{2022}]{ShyamalkumarTao2022}
Shyamalkumar, N. and S.~Tao (2022).
\newblock t-copula from the viewpoint of tail dependence matrices.
\newblock {\em J. Multivar. Anal.\/}~{\em 191}, 105027.

\bibitem[\protect\citeauthoryear{Sordo, Bello, and Su{\'a}rez-Llorens}{Sordo
  et~al.}{2018}]{sordo2018stochastic}
Sordo, M.~A., A.~J. Bello, and A.~Su{\'a}rez-Llorens (2018).
\newblock Stochastic orders and co-risk measures under positive dependence.
\newblock {\em Insurance Math. Econom.\/}~{\em 78}, 105--113.

\bibitem[\protect\citeauthoryear{Sordo, Su{\'a}rez-Llorens, and Bello}{Sordo
  et~al.}{2015}]{sordo2015comparison}
Sordo, M.~A., A.~Su{\'a}rez-Llorens, and A.~J. Bello (2015).
\newblock Comparison of conditional distributions in portfolios of dependent
  risks.
\newblock {\em Insurance Math. Econom.\/}~{\em 61}, 62--69.

\bibitem[\protect\citeauthoryear{Tobias and Brunnermeier}{Tobias and
  Brunnermeier}{2016}]{tobias2016covar}
Tobias, A. and M.~K. Brunnermeier (2016).
\newblock Covar.
\newblock {\em Amer. Econ. Rev.\/}~{\em 106\/}(7), 1705.

\bibitem[\protect\citeauthoryear{White}{White}{1982}]{white1982maximum}
White, H. (1982).
\newblock Maximum likelihood estimation of misspecified models.
\newblock {\em Econometrica\/}, 1--25.

\end{thebibliography}

\end{document}